\documentclass[conference,letterpaper]{IEEEtran}
\usepackage[utf8]{inputenc} 
\usepackage[T1]{fontenc}
\usepackage{url}
\usepackage{cite}
\usepackage[cmex10]{amsmath} 
\usepackage{booktabs}       
\usepackage{nicefrac}       
\usepackage{graphicx} 
\usepackage{amssymb,amsthm} 

\usepackage{subcaption}

\newtheorem{Theorem}{Theorem}
\newtheorem{Lemma}{Lemma}
\newtheorem{Corollary}{Corollary}
\newtheorem{Remark}{Remark}

\DeclareMathOperator*{\argmax}{argmax}

\DeclareMathOperator*{\Bern}{Bern}
\DeclareMathOperator*{\Gammadis}{Gamma}
\DeclareMathOperator*{\Exp}{\textit{Exp}}
\newcommand{\bad}{{\mathcal B}}
\DeclareMathOperator*{\Gdiv}{D_g}
\DeclareMathOperator*{\Ediv}{D_e}
\DeclareMathOperator*{\diag}{diag}

\def\Dist{P}
\def\Prob{{\mathbb P}}

\def\muvec{{\boldsymbol\mu}}
\def\sigvec{{\boldsymbol\Sigma}}
\def\lamvec{\boldsymbol\lambda}
\def\tetvec{\boldsymbol\theta}

\def\p{{\boldsymbol p}}
\def\w{{\boldsymbol w }}
\def\W{{\boldsymbol W }}
\def\mutilde{\Tilde{{\boldsymbol \mu}}}
\def\muhat{\hat{{\boldsymbol \mu}}}
\def\Sigmatilde{\Tilde{{\boldsymbol \Sigma}}}
\def\Sigmahat{\hat{{\boldsymbol \Sigma}}}
\def\rhotilde{\Tilde{{ \rho}}}
\def\rhohat{\hat{{ \rho}}}
\def\lambdatilde{\Tilde{{\boldsymbol \lambda}}}
\def\lambdahat{\hat{{\boldsymbol \lambda}}}

\def\taustar{{\tau}^*}

\newcommand{\sigmascalar}[2]{{\Sigma_{#1#2}}}
\newcommand{\muscalar}[2]{{\mu_{#1#2}}}
\newcommand{\lambdascalar}[2]{{\lambda_{#1#2}}}

\begin{document}
\title{Community Detection: Exact Recovery in Weighted Graphs} 

\author{%
  \IEEEauthorblockN{Mohammad~Esmaeili and Aria~Nosratinia}
  \IEEEauthorblockA{Department of Electrical and Computer Engineering, The University of Texas at Dallas\\
  	Email: \{Esmaeili, Aria\}@utdallas.edu}
}

\maketitle

\begin{abstract}
In community detection, the exact recovery of communities (clusters) has been mainly investigated under the general stochastic block model with edges drawn from Bernoulli distributions. This paper considers the exact recovery of communities in a complete graph in which the graph edges are drawn from either a set of Gaussian distributions with community-dependent means and variances, or a set of exponential distributions with community-dependent means. For each case, we introduce a new semi-metric that describes sufficient and necessary conditions of exact recovery. The necessary and sufficient conditions are asymptotically tight. The analysis is also extended to incomplete, fully connected weighted graphs.
\end{abstract}


\section{Introduction}

A main thrust of community detection literature has been on the stochastic block model with the graph edges drawn from Bernoulli distributions~\cite{abbe2015community, holland1983stochastic, hajek2015exact, saade2015spectral, esmaeili2019community, fronczak2013exponential, 9174105}, under various recovery metrics~\cite{decelle2011inference, mossel2015reconstruction, massoulie2014community, mossel2018proof, saad2016exit, mossel2016density, yun2014community, abbe2017community, abbe2015exact, mossel2015consistency}, and algorithms~\cite{chen2016statistical, mossel2014belief, 8682223, amini2018semidefinite, hajek2016achieving}. Exact recovery threshold of general stochastic block model was derived in~\cite{abbe2015community} by approximating Binomial distributions by Poisson distributions and utilizing the Chernoff-Hellinger divergence. 

While binary edges represent several practical applications and are analytically more tractable, there are many real-world graphs in which edge weights are better modelled by continuous values. For example,
brain networks are intrinsically weighted, reflecting a continuous distribution of connectivity strengths between different brain regions~\cite{nicolini2017community}. Applications in communications, e.g., data forwarding in Delay Tolerant Networks (DTN) and worm containment in Online Social Networks (OSN)~\cite{lu2014algorithms} also are well represented with continuous-valued weighted graphs. 
The edges of social media networks can be of different types, such as simple, weighted, directed and multi-way (i.e. connecting more than two entities) depending on the network creation process~\cite{papadopoulos2012community}.
In biology, community detection is applied on weighted gene networks for revealing cancers and anomalous tissues~\cite{cantini2015detection}.
For these applications, the stochastic block model with continuous probability density functions such as Gaussian distributions is the more appropriate choice.

For community detection from continuous-valued weighted graphs, only a few information-theoretic results are known, mostly under Gaussian distributions.
In~\cite{7820122},  weak recovery and exact recovery of a hidden community is investigated while the edges are drawn from two different Gaussian distributions. This is a symmetric version of the submatrix localization (also known as noisy biclustering) problem~\cite{refId0, kolar2011minimax, chen2016statistical}. 
In submatrix localization problem, the task is to detect a small block (blocks) with atypical mean within a large Gaussian matrix.
Binary symmetric communities with Gaussian distributions are investigated in~\cite{wu2018statistical}.
The problem of detecting a sparse principal component based on a sample from a multivariate Gaussian distribution in high dimensions is considered in~\cite{berthet2013optimal}.

Community detection in a more general setting similar to~\cite{abbe2015community} and under well-known continuous probability density functions is an interesting and challenging problem from both algorithmic and information-theoretic perspectives. This paper investigates this problem and obtains information limits for exact recovery of communities. 

The contributions of this paper are as follows. First, we analyze the exact recovery of node labels in a complete graph in which the edge weights are drawn from either a set of Gaussian or a set of Exponential distributions whose parameters are determined by the latent labels.  Under this model, sufficient and necessary conditions for exact recovery are derived.
Second, we extend the results to fully-connected but incomplete weighted graphs, by showing that under some conditions the inter and intra community probability distributions can be approximated by Gaussian distributions. 
The contributions of this paper and techniques that are used here  are widely applicable for other high-dimensional inference problems such as sparse PCA, Gaussian mixture clustering, tensor PCA, and other community detection problems with continuous distributions. 

\section{System Model \& Main Results}

Notation: $\Prob$  indicates the probability operator and $\Dist$ a probability distribution which is identified by the choice of its variables whenever there is no confusion. A matrix $\boldsymbol A$ has columns $\boldsymbol A_{i}$ and elements $A_{ij}$.
$\mathbb{R}$ is the set of real numbers, $\mathbb{R}_{+}$ is the set of non-negative real numbers, and $\mathbb{R}_{++}$ is the set of positive real numbers. 


We start by considering a complete graph with $n$ nodes. The graph nodes are divided into $K$ communities, where $K$ is finite.
Let $ \boldsymbol Q$ be an $K\times K$ matrix with entries $Q_{ij}$.
A node from community $i$ is connected to a node in community $j$ by a weighted edge drawn from distribution $Q_{ij}$.
In this paper, $Q_{ij}$ belongs to either a set of Gaussian or a set of Exponential distributions.
Let $\p \triangleq  [p_1, p_2, \cdots, p_K]$, where $p_i$ denote the size of community $i$. It is assumed that the size of each community is proportional to $n$, i.e., $p_i = \left \lfloor  \rho_i n \right \rfloor$, where $\rho_i \in (0,1)$ and $\sum_{i=1}^{K} \rho_i = 1$. 

When $Q_{ij}$ belongs to the set of Gaussian distributions, $Q_{ij} = \mathcal{N}(\bar{\mu}_{ij}, \bar{\sigma}_{ij}^2)$.
For this case, we define  matrices $\muvec$ and $\sigvec $ with entries $\mu_{ij} = p_i \bar{\mu}_{ij}$ and $\Sigma_{ij} = p_i \bar{\sigma}_{ij}^2$, respectively. 
When $Q_{ij}$ belongs to the set of Exponential distributions, $Q_{ij} = \Exp(\Tilde{\lambda}_{ij})$.
For this case, we define  matrix $\lamvec $ with entries $\lambda_{ij} = \Tilde{\lambda}_{ij}$.

Under the model with Gaussian distributions, assume that each edge is removed by a Bernoulli random variable. Then an edge from a node in community $i$ to a node in community $j$ is removed with probability $1-\theta_{ij}$. To have a fully connected graph, we consider a regime in which $\theta_{ij} = c_{ij} \frac{\log n}{n}$, where $c_{ij}$ is a constant. For this case,  we define matrix $\tetvec $ with entries $\theta_{ij}$. In this paper, this model is called  incomplete but fully connected weighted graph with Gaussian distributions.


Now, we summarize the main results  of this paper.
For convenience define the following semi-metrics:
\begin{align*}
    \Gdiv(\mutilde, \muhat, \Sigmatilde, \Sigmahat) \triangleq \max_{t \in [0,1]}   &\sum_{k=1}^{K} \Bigg\{  \frac{\Tilde{\mu}_{k}^{2}\hat{\Sigma}_{k}t+\hat{\mu}_{k}^{2} \Tilde{\Sigma}_{k} (1-t)}{2 \Tilde{\Sigma}_{k} \hat{\Sigma}_{k}} \\
    &-\frac{\left[ \Tilde{\mu}_{k}\hat{\Sigma}_{k}t+\hat{\mu}_{k}\Tilde{\Sigma}_{k} (1-t) \right ]^{2}}{2 \Tilde{\Sigma}_{k} \hat{\Sigma}_{k} \left[\hat{\Sigma}_{k}t+\Tilde{\Sigma}_{k} (1-t) \right ]} \\
    &-\frac{1}{2}\log \left(\frac{\Tilde{\Sigma}_{k}^{1-t} \hat{\Sigma}_{k}^{t} }{\Tilde{\Sigma}_{k}(1-t) +\hat{\Sigma}_{k}t} \right ) \Bigg \} ,
\end{align*}
\begin{align*}
    \Ediv(\lambdatilde, \lambdahat, \p) \triangleq& \max_{t \in [0,1]} \sum_{k=1}^{K} p_k \log \left( \frac{(1-t)\Tilde{\lambda}_{k}+t\hat{\lambda}_{k}}{\Tilde{\lambda}_{k}^{1-t} \hat{\lambda}_{k}^{t}} \right) .
\end{align*}
\begin{Theorem}
\label{Theorem-1}
With Gaussian distributions, 
\begin{itemize}
    \item 
    when $\Gdiv(\muvec_{i}, \muvec_{j}, \sigvec_{i}, \sigvec_{j}) = \omega(\log n)$ exact recovery of node labels is possible if and only if
    \begin{align*}
    \min_{i,j, i\neq j} \Gdiv(\muvec_{i}, \muvec_{j}, \sigvec_{i}, \sigvec_{j}) > 0 .
    \end{align*}
    
    \item 
    when $\Gdiv(\muvec_{i}, \muvec_{j}, \sigvec_{i}, \sigvec_{j}) = O(\log n)$ exact recovery of node labels is possible if and only if
    \begin{align*}
    \min_{i,j, i\neq j} \lim_{n \rightarrow \infty} \frac{\Gdiv(\muvec_{i}, \muvec_{j}, \sigvec_{i}, \sigvec_{j})}{\log n}> 1  .
    \end{align*}
\end{itemize}
\end{Theorem}

\begin{Theorem}
\label{Theorem-2}
With Exponential distributions,
\begin{itemize}
    \item 
    when $\Ediv(\lamvec_{i}, \lamvec_{j}, \p) = \omega(\log n)$, exact recovery of node labels is possible if and only if
    \begin{align*}
    \min_{i,j, i\neq j} \Ediv(\lamvec_{i}, \lamvec_{j}, \p) > 0 .
    \end{align*}
    
    \item 
    when $\Ediv(\lamvec_{i}, \lamvec_{j}, \p)) = O(\log n)$, exact recovery of node labels is possible if and only if
    \begin{align*}
    \min_{i,j, i\neq j} \lim_{n \rightarrow \infty} \frac{\Ediv(\lamvec_{i}, \lamvec_{j}, \p)}{\log n}> 1  .
    \end{align*}
\end{itemize}
\end{Theorem}

\begin{Remark}
For both the Gaussian and the Exponential cases, when the related semi-metric is $\omega(\log n)$, the exact recovery condition is equivalent to  
\[
\boldsymbol Q_i \neq \boldsymbol Q_j   \qquad \forall i\neq j . 
\]
\end{Remark}

\begin{Corollary}
For a fully connected weighted but {\em incomplete} graph whose edge weights are Gaussian distributed, exact recovery of node labels is possible if and only if
\begin{align*}
\min_{i,j, i\neq j} \lim_{n \rightarrow \infty} \frac{\Gdiv(\muvec_{i}, \muvec_{j}, \sigvec_{i}, \sigvec_{j})}{\log n}> 1  ,
\end{align*}
where 
\begin{alignat*}{2}
    & \mu_{ij} = p_i \bar{\mu}_{ij} \theta_{ij} \qquad &\forall i,j ,\\
    &\Sigma_{ij} = p_i \theta_{ij} [\bar{\sigma}_{ij}^2+(1-\theta_{ij})\bar{\mu}_{ij}^2]  \qquad &\forall i,j.
\end{alignat*}
\end{Corollary}

\section{Proofs}

At each node, our problem is equivalent to testing a hypothesis $H$ indicating which community the node belongs to, out of the set of $K$ communities. 
In our setting, this is a Bayesian problem with prior $\mathbb P(H=i) = \rho_i$. For each node, let $\W$ be a random vector with entries $W_i$ representing the summation of edge weights connecting a node of interest to nodes in community $i$.



Assume that all node labels are revealed except for one, whose community membership $H$ is to be derived based on an observation of $\W$. The maximum a posteriori estimator (MAP) is
\begin{align*}
     \argmax_{i} \rho_i \Dist(\w | H=i)  . 
\end{align*}
A simple comparison can eliminate a candidate, i.e, if
\begin{align}
\label{equ:1}
     \rho_i \Dist(\w | H=i) < \rho_k \Dist(\w | H=k) ,
\end{align}
then $H \neq i$. Therefore, a set of pairwise comparisons of the hypotheses reveals the MAP. Assume that the true hypothesis is $H = i$. Denote by $\bad(i,k)$ the region of $\W$ for which \eqref{equ:1} is satisfied, i.e.,  $H=i$ has a worse metric compared with $H = k$. Also denote by $\bad(i)$ the region for $\W$ where the overall MAP estimator is in error. Then the probability of error is
\begin{align*}
    P_e = &\sum_{i } \rho_i \Prob(\W \in \bad(i) | H=i) .
\end{align*}
Since $\bad(i) \subset \bigcup_{k=1}^{K} \bad(i,k)$, 
\begin{align}
\label{equ:4}
    P_e  \leq \sum_{i } \sum_{k, k \neq i} \rho_i \Prob(\W \in \bad(i,k) | H=i) .
\end{align}
Define
\begin{align*}
    I(\w,i,k) \triangleq \min &\{ \rho_i\Dist(\w | H=i) , \rho_k \Dist(\w | H=k) \},
\end{align*}
and note that 
\begin{align}
\label{equ:3}
    I(\w,i,k) =\begin{cases}
    \rho_i \Dist(\w | H=i)  & \text{when} ~\W \in \bad(i,k) \\
    \rho_k \Dist(\w | H=k)   & \text{when} ~\W \in  \bad^c(i,k)
    \end{cases} .
\end{align}
Substituting \eqref{equ:3} into \eqref{equ:4}, 
\begin{align}
\label{equ:2}
    P_e  \leq \int \sum_{i } \sum_{k > i}I(\w,i,k) \text{d}\w.
\end{align}
The error is bounded from below by
\begin{align}
    \label{equ:5}
    P_e  \geq \frac{1}{K-1} 
    \int \sum_{i } \sum_{k > i} I(\w,i,k) \text{d}\w , 
\end{align}
because
\begin{align*}
    \sum_{\substack{k  \\ k \neq i}} \Prob(\W \in \bad(i,k) | H=i) \leq (K-1) \Prob(\W \in \bad(i) | H=i) .
\end{align*}
Therefore, the error probability is bounded by controlling
\begin{align}
\label{equ:6}
    \int I(\w,i,k) \, \text{d}\w.
\end{align}

\subsection{Proof of Theorem~\ref{Theorem-1}}
For a node in community $i$, the edge sums $W_j$ are distributed according to $\mathcal{N}(p_j \bar{\mu}_{ji}, p_j  \bar{\sigma}_{ji}^2)$, and are independent of each other. We collect these edge sums into the vector $\W$, which obeys a multivariate Gaussian distribution with mean denoted $\muvec_{i}$ and covariance matrix  $\diag (\sigvec_{i})$. Then
\begin{align*}
    f(\w; \muvec_{i}, \sigvec_{i}) &\triangleq \Dist\left ( \w| H = i\right) \\
    &= \prod_{k=1}^{K} \frac{1}{\sqrt{2\pi \sigmascalar{k}{i}}} \exp \left( -\frac{(w_{k}- \muscalar{k}{i})^2}{2 \sigmascalar{k}{i}} \right), 
\end{align*}
where $\muscalar{k}{i} = p_k \bar{\mu}_{ki}$ and $\sigmascalar{k}{i} = p_k  \bar{\sigma}_{ki}^{2}$.

\begin{Lemma}
\label{lemma: Gaussian-div}
Let
$\mutilde, \muhat \in \mathbb{R}^{K}$,  $\Sigmatilde, \Sigmahat \in \mathbb{R}_{++}^K$, and $\rhotilde, \rhohat \in \mathbb{R}_{++}$.
If either $\mutilde \neq \muhat$ or $\Sigmatilde \neq \Sigmahat$, then 
\begin{align*}
    \int_{\mathbb{R}^{K}} &\min \{ \rhotilde \,f(\w; {\mutilde}, {\Sigmatilde}) , \rhohat\, f(\w; \muhat, {\Sigmahat}) \} \text{d}\w \\
    & \leq e^{-\Gdiv(\mutilde,\muhat, \Sigmatilde, \Sigmahat) + c_1} , \\
    \int_{\mathbb{R}^{K}} &\min \{ \rhotilde f(\w; \mutilde, \Sigmatilde) , \rhohat f(\w; \muhat, \Sigmahat) \} \text{d}\w \\
    &\geq  e^{-\Gdiv(\mutilde,\muhat, \Sigmatilde,
    \Sigmahat) +c_2 }  ,
\end{align*}
where $c_1$ and $c_2$ are some constants.
\end{Lemma}
\begin{proof}
Define 
\begin{align*}
    g_{1}(t) &\triangleq \left( \frac{f(\w; {\mutilde}, {\Sigmatilde})}{f(\w; \muhat, {\Sigmahat})}\right)^{1-t},\\
    g_{2}(t) & \triangleq \left( \frac{f(\w; \muhat, {\Sigmahat})}{f(\w; {\mutilde}, {\Sigmatilde})}\right)^{t},\\
    g(t) &\triangleq  f(\w; {\mutilde}, {\Sigmatilde})^{t} f(\w; \muhat, {\Sigmahat})^{1-t} ,
\end{align*}
in which the dependence of $g_1(t)$, $g_2(t)$, and $g(t)$ on $\w$ is suppressed for notational convenience. 
Note that $g(t)$ can be restated as
\begin{align*}
    g(t) =  e^{-\sum_{k=1}^{K}D_k(t)} \prod_{k=1}^{K}  \frac{1}{\sqrt{2\pi \sigma_k^2(t) }}
    \exp \left( -\frac{(w_{k}-\mu_k(t))^2}{2\sigma_k^2(t)} \right), 
\end{align*}
where 
\begin{align*}
    \mu_k(t) \triangleq& \frac{\Tilde{\mu}_{k}\hat{\Sigma}_{k}t+\hat{\mu}_{k}\Tilde{\Sigma}_{k} (1-t)}{\hat{\Sigma}_{k}t+\Tilde{\Sigma}_{k} (1-t)} , \\
    \sigma_k^2(t) \triangleq& \frac{\Tilde{\Sigma}_{k} \hat{\Sigma}_{k}}{\hat{\Sigma}_{k}t+\Tilde{\Sigma}_{k} (1-t)} ,\\
    D_k(t) \triangleq& \frac{\Tilde{\mu}_{k}^{2}\hat{\Sigma}_{k}t+\hat{\mu}_{k}^{2} \Tilde{\Sigma}_{k} (1-t)}{2 \Tilde{\Sigma}_{k} \hat{\Sigma}_{k}} -\frac{\left[ \Tilde{\mu}_{k}\hat{\Sigma}_{k}t+\hat{\mu}_{k}\Tilde{\Sigma}_{k} (1-t) \right ]^{2}}{2 \Tilde{\Sigma}_{k} \hat{\Sigma}_{k} \left[\hat{\Sigma}_{k}t+\Tilde{\Sigma}_{k} (1-t) \right ]} \\
    &-\frac{1}{2}\log \left(\frac{\Tilde{\Sigma}_{k}^{1-t} \hat{\Sigma}_{k}^{t} }{\Tilde{\Sigma}_{k}(1-t) +\hat{\Sigma}_{k}t} \right ).
\end{align*}

\begin{Lemma}
\label{lemma2}
For any $t\in [0,1]$, $\min \{ g_{1}(t), g_{2}(t) \} \leq 1$. 
\end{Lemma}
\begin{proof}
Both $g_{1}(t)$ and $g_{2}(t)$ are monotonic and $\frac{g_2(t)}{g_1(t)}$ is a positive constant (does not depend on $t$), thus $\min\{g_1(t),g_2(t)\}$ is also monotonic in $t$. Since $g_1(1)=g_2(0)=1$, for all $t$ we have:
\begin{align*}
\min \{ g_{1}(t), g_{2}(t) \} \leq 1. 
\end{align*}
\end{proof}

It can be shown that for any $t \in [0,1]$, 
\begin{align*}
    \int_{\mathbb{R}^{K}} &\min \{ \rhotilde f(\w; \mutilde, \Sigmatilde) , \rhohat f(\w; \muhat, \Sigmahat) \} \text{d}\w \\
    &\leq \max \{\rhotilde, \rhohat\} \int_{\mathbb{R}^{K}} \min \{ f(\w; \mutilde, \Sigmatilde), f(\w; \muhat, \Sigmahat)\} \text{d}\w \\
    &= \max \{\rhotilde, \rhohat\} \int_{\mathbb{R}^{K}} g(t) \min \{g_1(t), g_2(t)\} \text{d}\w \\
    &\leq \max \{\rhotilde, \rhohat\} e^{-\sum_{k=1}^{K}D_k(t)} ,
\end{align*}
where the last inequality holds due to Lemma~\ref{lemma2} and 
\begin{align*}
    \int_{\mathbb{R}^{K}} \prod_{k=1}^{K}  \frac{1}{\sqrt{2\pi \sigma_k^2(t) }}
    \exp \left( -\frac{(w_{k}-\mu_k(t))^2}{2\sigma_k^2(t)} \right) \text{d}\w =1 .
\end{align*}
When $t$ is chosen to minimize $e^{-\sum_{k=1}^{K}D_k(t)}$, 
\begin{align*}
    \int_{\mathbb{R}^{K}} &\min \{ \rhotilde f(\w; \mutilde, \Sigmatilde) , \rhohat f(\w; \muhat, \Sigmahat) \} \text{d}\w \\
    &\leq \max \{\rhotilde, \rhohat\} e^{-\Gdiv(\mutilde,\muhat, \Sigmatilde, \Sigmahat)} .
\end{align*}

To prove the second half, note that 
\begin{align}
\label{equ:7}
    \min \{g_1(t^*), g_2(t^*)\} = g_2(\taustar), 
\end{align}
where $\taustar \triangleq t^*$ if $\min \{g_1(t^*), g_2(t^*)\} = g_2(t^*)$; Otherwise $\taustar \triangleq  t^*-1$. Hence, at $t^*$,
\begin{align*}
    \int_{\mathbb{R}^{K}} &\min \{ \rhotilde f(\w; \mutilde, \Sigmatilde) , \rhohat f(\w; \muhat, \Sigmahat) \} \text{d}\w \\
    \geq& \min \{\rhotilde, \rhohat\} \int_{\mathbb{R}^{K}} \min \{ f(\w; \mutilde, \Sigmatilde), f(\w; \muhat, \Sigmahat)\} \text{d}\w \\
    =& \min \{\rhotilde, \rhohat\} \int_{\mathbb{R}^{K}} g(t^*) \min \{g_1(t^*), g_2(t^*)\} \text{d}\w \\
    =& \min \{\rhotilde, \rhohat\} e^{-\Gdiv(\mutilde,\muhat, \Sigmatilde, 
    \Sigmahat)} \\
    &\times \int_{\mathbb{R}^{K}} \prod_{k=1}^{K}   \frac{1}{\sqrt{2\pi \sigma_k^2(t^*)}} e^{ - \frac{(w_{k}-\mu_k(t^*))^2}{2\sigma_k^2(t^*)}  } g_2(w_k, \taustar) \text{d}\w ,
\end{align*}
where 
\begin{align*}
    g_2(w_k, \taustar)  = \left ( \frac{\Tilde{\Sigma}_k}{\hat{\Sigma}_k} \right)^{\frac{\taustar}{2}} e^{-\left[ \frac{(w_{k}-\hat{\mu}_k)^{2}}{2\hat{\Sigma}_k} -\frac{(w_{k}-\Tilde{\mu}_k)^{2}}{2\Tilde{\Sigma}_k} \right ] \taustar } . 
\end{align*}
Since $g_2(w_k, \taustar)$ is a non-negative and integrable function of $w_k$, applying a generalized variant of the mean value Theorem, there exists $w_k^*$ such that
\begin{align*}
    \int_{\mathbb{R}}  \frac{1}{\sqrt{2\pi \sigma_k^2(t^*)}} e^{ - \frac{(w_{k}-\mu_k(t^*))^2}{2\sigma_k^2(t^*)}  } g_2(w_k, \taustar) \text{d}w_k = g_2(w_k^*, \taustar) .
\end{align*}
It can be shown that at $\taustar$, $g_2(w_k^*, \taustar)$ is a positive constant. 
Therefore, 
\begin{align*}
    \int_{\mathbb{R}^{K}} &\min \{ \rhotilde f(\w; \mutilde, \Sigmatilde) , \rhohat f(\w; \muhat, \Sigmahat) \} \text{d}\w \\
    \geq& \min \{\rhotilde, \rhohat\} e^{-\Gdiv(\mutilde,\muhat, \Sigmatilde,
    \Sigmahat) +c } ,
\end{align*}
where $c$ is a constant. 
\end{proof}
Using Lemma~\ref{lemma: Gaussian-div} and the bounds~\eqref{equ:2} and~\eqref{equ:5}, 
\begin{align*}
    &P_e \leq e^{-\Gdiv(\muvec_{i},\muvec_{j}, \sigvec_{i}, \sigvec_{j}) + c_1} , \\
    &P_e \geq e^{-\Gdiv(\muvec_{i},\muvec_{j}, \sigvec_{i},
    \sigvec_{j}) +c_2 }.
\end{align*}
When $\Gdiv(\muvec_{i}, \muvec_{j}, \sigvec_{i}, \sigvec_{j}) = \omega(\log n)$, as $n$ goes to infinity, exact recovery is possible if and only if
\begin{align*}
\min_{i,j, i\neq j} \Gdiv(\muvec_{i}, \muvec_{j}, \sigvec_{i}, \sigvec_{j}) > 0 .
\end{align*}
If $\muvec_{i}$ is close to $\muvec_{j}$ and $\sigvec_{i}$ is close to $\sigvec_{j}$, then $\Gdiv(\muvec_{i}, \muvec_{j}, \sigvec_{i}, \sigvec_{j}) = O(\log n)$. In this regime, exact recovery is possible if and only if
\begin{align*}
\min_{i,j, i\neq j} \lim_{n \rightarrow \infty} \frac{\Gdiv(\muvec_{i}, \muvec_{j}, \sigvec_{i}, \sigvec_{j})}{\log n}> 1  .
\end{align*}

\subsection{Proof of Theorem~\ref{Theorem-2}}
If the node of interest belongs to community $i $, $W_j$ is distributed according to $\Gammadis(p_j,\lambda_{ji})$. The vector $\boldsymbol W$ has independent Gamma entries with different means $p_j/\lambda_{ji}$. Under $H=i$, random variable $W$ is drawn from a multivariate Gamma distribution with shape parameter $\p$ and rate parameter $\lamvec_{i} \in \mathbb{R}_{++}^{K}$. Then
\begin{align*}
    f(\w; \p, \lamvec_{i}) \triangleq \Dist \left ( \w| H = i\right) &=  \prod_{k=1}^{K} \frac{\lambdascalar{k}{i}^{p_k}}{\Gamma(p_k)} w^{p_k-1}_{k} e^{-\lambdascalar{k}{i}w_{k}}.
\end{align*}

\begin{Lemma}
\label{lemma: Exponential-div}
Let
$\lambdatilde, \lambdahat \in \mathbb{R}_{++}^{K}$, $\p \in \mathbb{R}_{++}^K$, and $\rhotilde, \rhohat \in \mathbb{R}_{++}$.
If $\lambdatilde \neq \lambdahat$,
\begin{align*}
    \int_{\mathbb{R}^{K}_{+}} &\min \{ \rhotilde f(\w; \p, \lambdatilde) , \rhohat f(\w; \p, \lambdahat) \} \text{d}\w \leq  e^{-\Ediv(\lambdatilde,\lambdahat, \p) +c_1 }, \\
    \int_{\mathbb{R}^{K}_{+}} &\min \{ \rhotilde f(\w; \p, \lambdatilde) , \rhohat f(\w; \p, \lambdahat) \} \text{d}\w \geq e^{-\Ediv(\lambdatilde,\lambdahat, \p) +c_2 } ,
\end{align*}
where $c_1$ and $c_2$ are some constants.
\end{Lemma}
\begin{proof}
Define 
\begin{align*}
    g_{1}(t) &\triangleq \left( \frac{f(\w; \p, \lambdatilde)}{f(\w; \p, \lambdahat)} \right)^{1-t} ,\\
    g_{2}(t) & \triangleq \left( \frac{f(\w; \p, \lambdahat)}{f(\w; \p, \lambdatilde)} \right)^{t} ,\\
    g(t) &\triangleq  f(\w; \p, \lambdatilde)^{t} f(\w; \p, \lambdahat)^{1-t} ,
\end{align*}
in which the dependence of $g_1(t)$, $g_2(t)$, and $g(t)$ on $\w$ is suppressed for notational convenience.
Notice that Lemma~\ref{lemma2} holds also in this case. 
For any $t \in [0,1]$, 
\begin{align*}
    \int_{\mathbb{R}^{K}_{+}} &\min \{ \rhotilde f(\w; \p, \lambdatilde) , \rhohat f(\w; \p, \lambdahat) \} \text{d}\w \\
    &\leq \max \{\rhotilde, \rhohat\} \int_{\mathbb{R}^{K}_{+}} \min \{ f(\w; \p, \lambdatilde), f(\w; \p, \lambdahat)\} \text{d}\w \\
    &= \max \{\rhotilde, \rhohat\} \int_{\mathbb{R}^{K}_{+}} g(t) \min \{g_1(t), g_2(t)\} \text{d}\w \\
    &\leq \max \{\rhotilde, \rhohat\} e^{-\sum_{k=1}^{K} p_k \log \left( \frac{(1-t)\Tilde{\lambda}_{k}+t\hat{\lambda}_{k}}{\Tilde{\lambda}_{k}^{1-t} \hat{\lambda}_{k}^{t}} \right) } ,
\end{align*}
where the last inequality holds due to Lemma~\ref{lemma2} and 
\begin{align*}
    \int_{\mathbb{R}_{+}^{K}} \prod_{k=1}^{K} \frac{ \left(\lambda_k(t) \right )^{p_k}}{\Gamma(p_k)} w^{p_k-1}_{k} e^{-w_{k} \lambda_k(t) }  \text{d}\w =1 ,
\end{align*}
where $ \lambda_k(t) \triangleq (1-t)\Tilde{\lambda}_{k}+t\hat{\lambda}_{k}$. 
When $t$ is chosen to maximize $\sum_{k=1}^{K} p_k \log \left( \frac{(1-t)\Tilde{\lambda}_{k}+t\hat{\lambda}_{k}}{\Tilde{\lambda}_{k}^{1-t} \hat{\lambda}_{k}^{t}} \right)$, 
\begin{align*}
    \int_{\mathbb{R}^{K}_{+}} &\min \{ \rhotilde f(\w; \p, \lambdatilde) , \rhohat f(\w; \p, \lambdahat) \} \text{d}\w \\
    &\leq \max \{\rhotilde, \rhohat\} e^{-\Ediv(\lambdatilde,\lambdahat, \p)}.
\end{align*}

Notice that~\eqref{equ:7} holds also in this case. 
Hence,  at $t^*$,
\begin{align*}
    \int_{\mathbb{R}^{K}_{+}} &\min \{ \rhotilde f(\w; \p, \lambdatilde) , \rhohat f(\w; \p, \lambdahat) \} \text{d}\w \\
    \geq& \min \{\rhotilde, \rhohat\} \int_{\mathbb{R}^{K}_{+}} \min \{ f(\w; \p, \lambdatilde), f(\w; \p, \lambdahat) \} \text{d}\w \\
    =& \min \{\rhotilde, \rhohat\} \int_{\mathbb{R}^{K}_{+}} g(t^*) \min \{g_1(t^*), g_2(t^*)\} \text{d}\w \\
    =& \min \{\rhotilde, \rhohat\} e^{-\Ediv(\lambdatilde,\lambdahat, \p)} \\
    &\times \int_{\mathbb{R}^{K}_{+}} \prod_{k=1}^{K}  \frac{\left( \lambda_k(t^*) \right )^{p_k}}{\Gamma(p_k)} (w_k)^{p_k-1} e^{-w_k \lambda_k(t^*) } g_2(w_k, \taustar) \text{d}\w , 
\end{align*}
where
\begin{align*}
    g_2(w_k, \taustar)  = \left ( \frac{\Tilde{\lambda}_k}{\hat{\lambda}_k} \right)^{ p_k \taustar} e^{-w_k \taustar (\Tilde{\lambda}_k-\hat{\lambda}_k)} . 
\end{align*}
Since $g_2(w_k, \taustar)$ is a non-negative and integrable function of $w_k$, applying a generalized variant of  mean value Theorem, there exists $w_k^*$ such that
\begin{align*}
    \int_{\mathbb{R}_{+}} \frac{1}{\sqrt{2\pi \sigma_k^2(t^*)}} e^{ - \frac{(w_{k}-\mu_k(t^*))^2}{2\sigma_k^2(t^*)}  } g_2(w_k, \taustar) \text{d}w_k = g_2(w_k^*, \taustar) .
\end{align*}
It can be shown that at $\taustar$, $g_2(w_k^*, \taustar)$ is a positive constant. 
Therefore,  
\begin{align*}
    \int_{\mathbb{R}^{K}_{+}} &\min \{ \rhotilde f(\w; \p, \lambdatilde) , \rhohat f(\w; \p, \lambdahat) \} \text{d}\w \\
    \geq& \min \{\rhotilde, \rhohat\} e^{-\Ediv(\lambdatilde,\lambdahat, \p) +c } ,
\end{align*}
where $c$ is a constant.
\end{proof}
Using Lemma~\ref{lemma: Exponential-div} and the bounds~\eqref{equ:2} and~\eqref{equ:5}, for some constants $c_1$ and $c_2$,
\begin{align*}
    &P_e \leq e^{-\Ediv(\lamvec_{i},\lamvec_{j}, \p) +c_1 } , \\
    &P_e \geq e^{-\Ediv(\lamvec_{i},\lamvec_{j}, \p)  +c_2 } .
\end{align*}
When $\Ediv(\lamvec_{i},\lamvec_{j}, \p) = \omega(\log n)$, as $n$ goes to infinity, exact recovery is possible if and only if
\begin{align*}
\min_{i,j, i\neq j} \Ediv(\lamvec_{i}, \lamvec_{j}, \p) > 0 .
\end{align*}
If $\lamvec_{i}$ is close to $\lamvec_{j}$, then $\Ediv(\lamvec_{i},\lamvec_{j}, \p) = O(\log n)$. In this regime, $\lim_{n \rightarrow \infty} \frac{\Gdiv(\muvec_{i}, \muvec_{j}, \sigvec_{i}, \sigvec_{j})}{\log n}$ is a constant and exact recovery is possible if and only if
\begin{align*}
\min_{i,j, i\neq j} \lim_{n \rightarrow \infty} \frac{\Ediv(\lamvec_{i}, \lamvec_{j}, \p)}{\log n}> 1  .
\end{align*}
\begin{figure}
\begin{center}
\begin{subfigure}{0.5\textwidth}
         \centering
         \includegraphics[width=\textwidth]{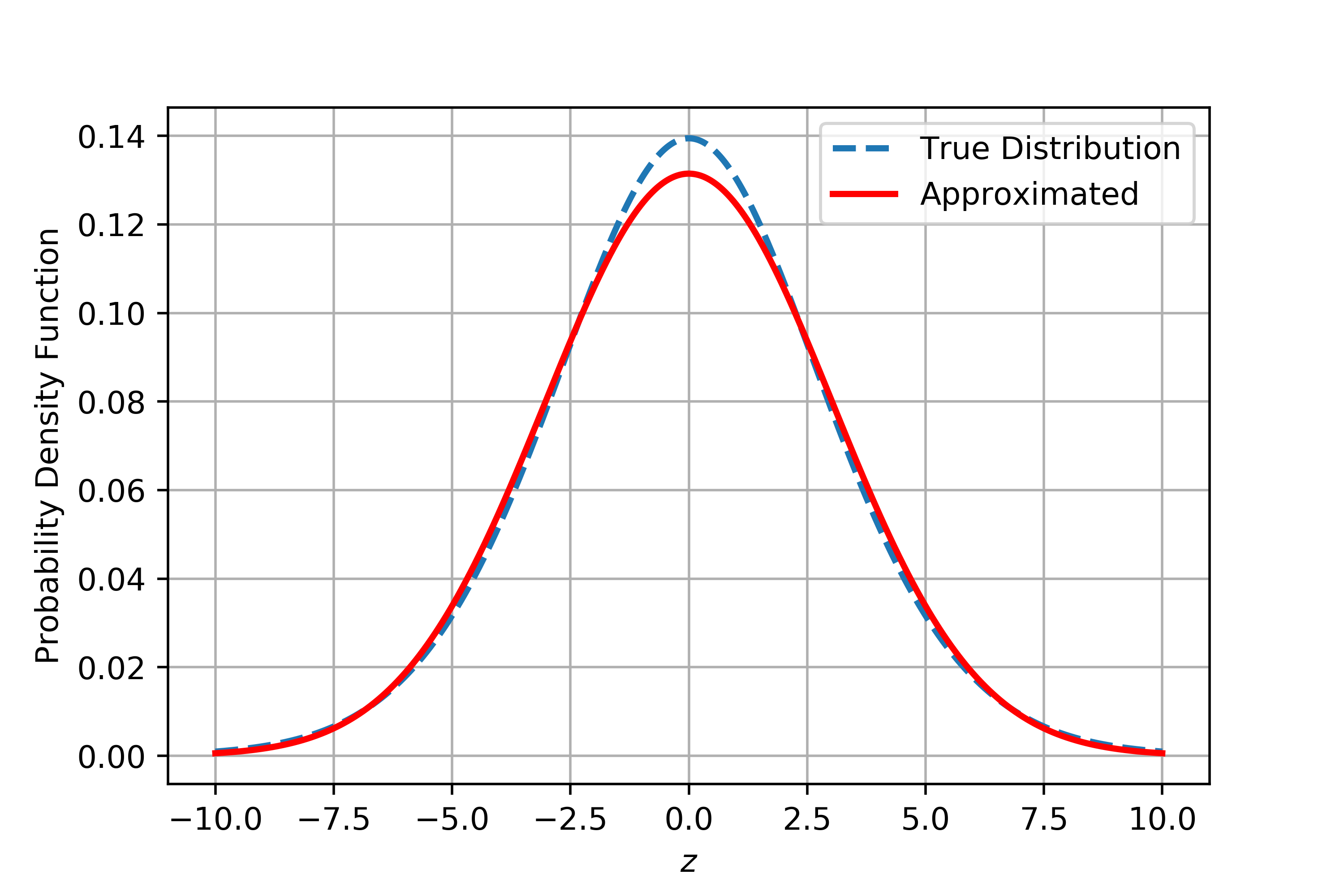}
         \caption{${\mu} = 0, {\sigma}^2 = 1$}
     \end{subfigure}
     \hfill
     \begin{subfigure}{0.5\textwidth}
         \centering
         \includegraphics[width=\textwidth]{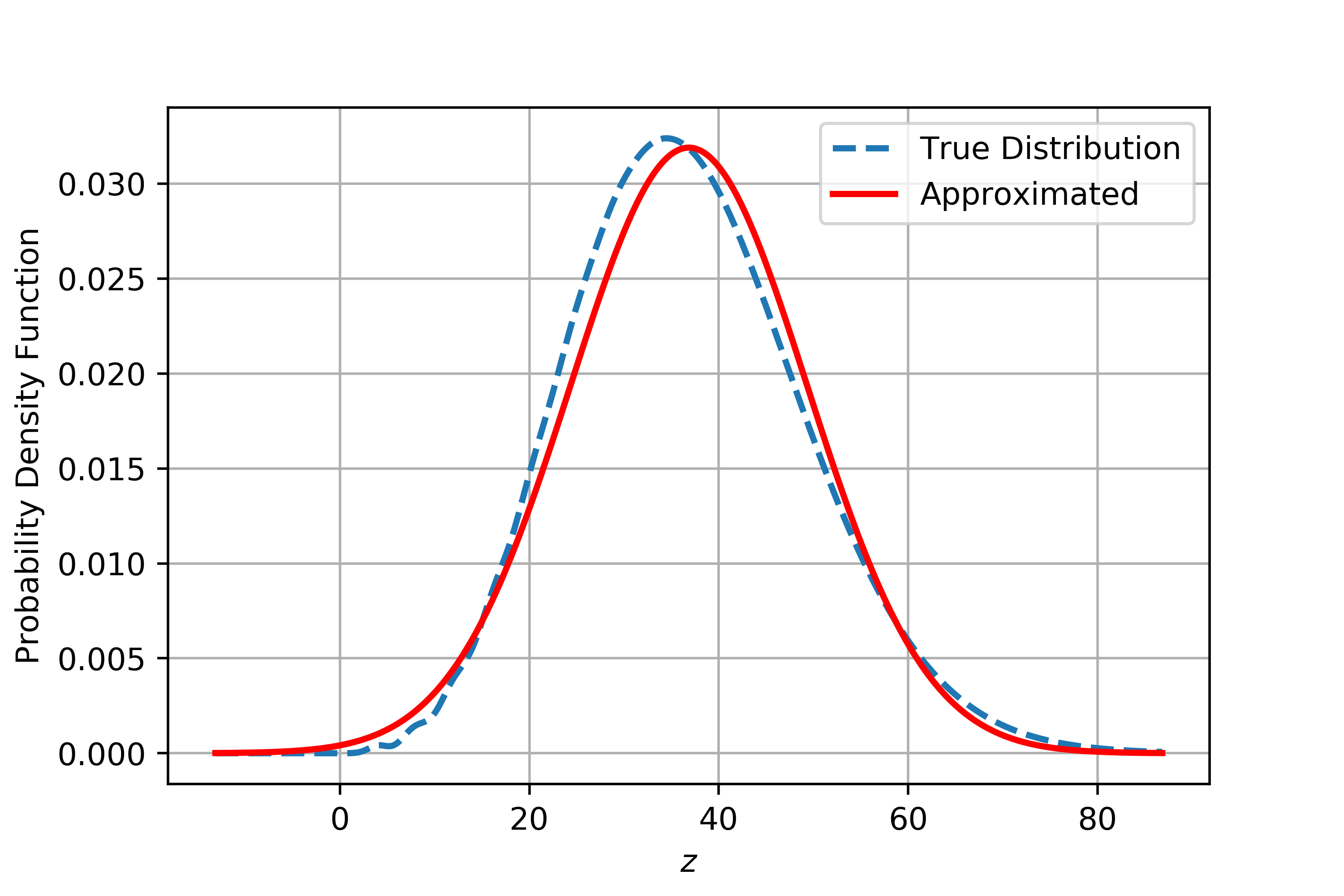}
         \caption{${\mu} = 4, {\sigma}^2 = 1$}
     \end{subfigure}
     \hfill
     \begin{subfigure}{0.5\textwidth}
         \centering
         \includegraphics[width=\textwidth]{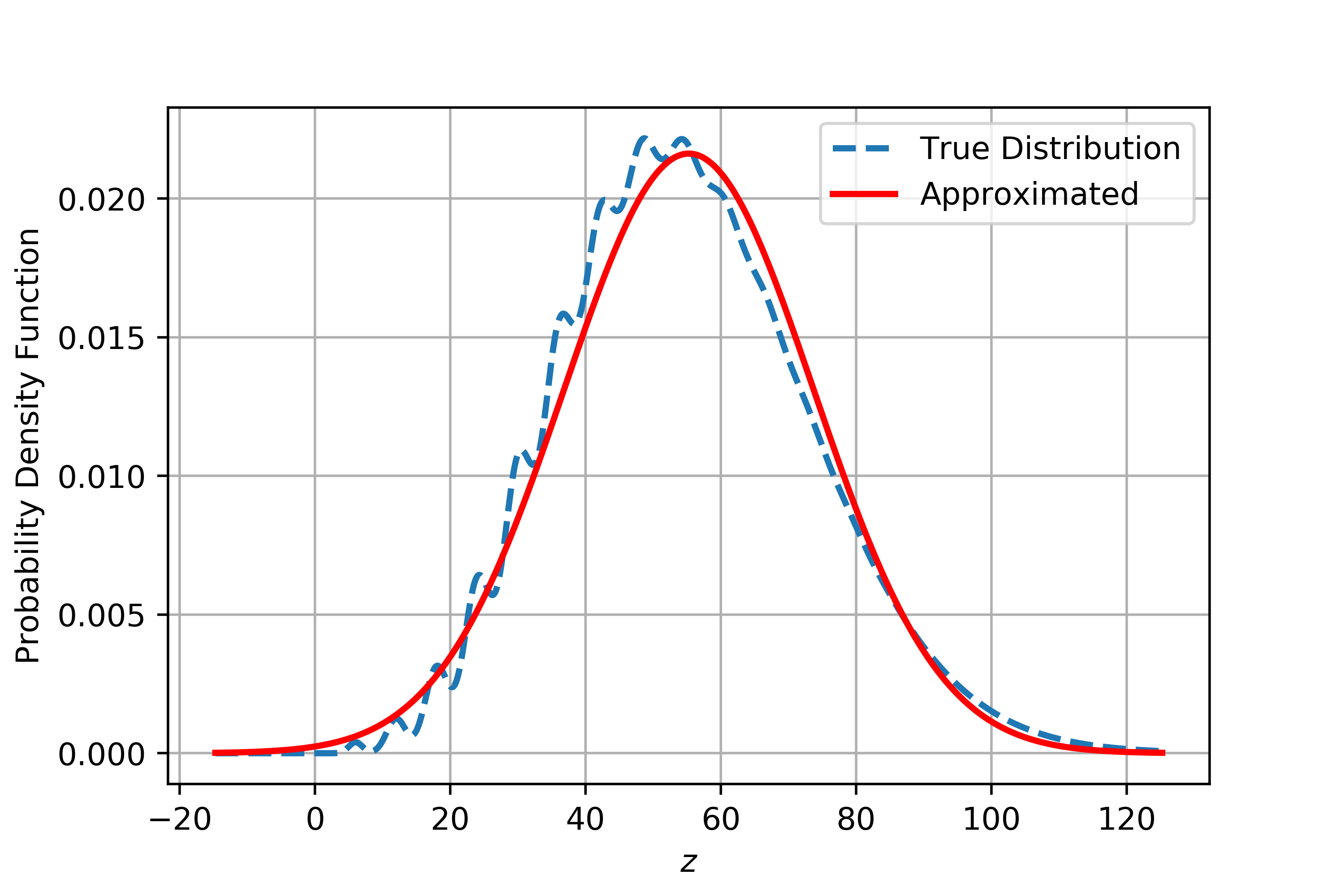}
         \caption{${\mu} = 6, {\sigma}^2 = 1$}
     \end{subfigure}
\end{center}
\caption{True distribution~\eqref{equ:8} and its  approximation for $\theta = \frac{\log n}{n}$, $n=10000$, and different values of $\frac{{\mu}}{{\sigma}}$.}
\label{fig1}
\end{figure}

\section{Incomplete but Fully Connected Weighted Graphs}
Let $X \sim \Bern(\theta)$ and $Y \sim \mathcal{N}(\mu,\sigma^2)$. Then $Z \triangleq XY$ is a random variable with probability density function
\begin{align*}
    f_Z(z) = \theta f_Y(z) + (1-\theta) \delta(z),
\end{align*}
where $f_Y(y)$ is the probability density function of $Y$ and $\delta(z)$ is Dirac delta function.
Then the probability density function of $ \sum_{i=1}^{n} Z_i$ is
\begin{align}
\label{equ:8}
      \sum_{i = 0}^{n} \binom{n}{i} \theta^i (1-\theta)^{n-i} \{ f_Y(z) \}^{\circledast i} \circledast \delta(z) ,
\end{align}
where $\circledast $ denotes the convolution operator.
In~\eqref{equ:8}, for each $i$, $\{ f_Y(z) \}^{\circledast i}$ is a Gaussian probability density function with mean $i\mu$ and variance of $i\sigma^2$. If $\theta$ is in order of $\frac{\log n}{n}$ and $\frac{\mu}{\sigma} \leq 4$, then
the probability density function~\eqref{equ:8} is well-enough approximated by a Gaussian distribution with mean $n \mu \theta$ and variance of $n\theta[\sigma^2+(1-\theta)\mu^2]$.
Figure~\ref{fig1} compares the probability density function~\eqref{equ:8} and its Gaussian approximation under the conditions mentioned above. 

Using this approximation and following Theorem~\ref{Theorem-1},
when $Q_{ij} = \mathcal{N}(\bar{\mu}_{ij}, \bar{\sigma}_{ij}^2)$,
exact recovery of node labels is possible if and only if
\begin{align*}
\min_{i,j, i\neq j} \lim_{n \rightarrow \infty} \frac{\Gdiv(\muvec_{i}, \muvec_{j}, \sigvec_{i}, \sigvec_{j})}{\log n}> 1  ,
\end{align*}
where
\begin{align*}
    & \mu_{ij}  = p_i \bar{\mu}_{ij} \theta_{ij} ,\\
    &\Sigma_{ij} = p_i \theta_{ij} [\bar{\sigma}_{ij}^2+(1-\theta_{ij})\bar{\mu}_{ij}^2] .
\end{align*}
\bibliographystyle{IEEEtran}
\bibliography{Ref}

\enlargethispage{-1.2cm} 
\IEEEtriggeratref{3}

\end{document}